%% file: final.tex
\documentclass{llncs}
\usepackage{latexsym,amsmath,amssymb,amsfonts,graphicx,mathrsfs,mathtools,txfonts}

\usepackage{float}     
\usepackage{import}   
\usepackage{times}        
\usepackage[usenames,dvipsnames]{color}
\usepackage{xcolor}
\usepackage{url}
\usepackage{verbatim}
\usepackage{bm}
\usepackage{rotating}
\usepackage{braket}

\newcommand{\cV}{\mathcal{V}}
\newcommand{\cN}{\mathcal{N}}
\newcommand{\cW}{\mathcal{W}}
\newcommand{\cE}{\mathcal{E}}
\newcommand{\cG}{\mathcal{G}}

\newcommand{\real}{\mathbb{R}}

\newcommand{\Pil}{\mathrm{\Pi}}

\makeatletter
\newcommand{\oset}[3][0ex]{%
  \mathrel{\mathop{#3}\limits^{
    \vbox to#1{\kern-2\ex@
    \hbox{$\scriptstyle#2$}\vss}}}}
\makeatother

\begin{document}

\frontmatter

\mainmatter

\title{Bounding the convergence time of local probabilistic evolution}
\author{Simon Apers\inst{1} \and Alain Sarlette\inst{1,2} \and Francesco Ticozzi\inst{3,4}}
\authorrunning{Simon Apers et al.}
\institute{Department of Electronics and Information Systems, Ghent University, Belgium
\and
QUANTIC lab, INRIA Paris, France
\and
Dipartimento di Ingegneria dell'Informazione, Universit\`a di Padova, Italy
\and
Department of Physics and Astronomy, Dartmouth College, NH 03755, USA\\
\email{simon.apers@ugent.be,alain.sarlette@inria.fr,ticozzi@dei.unipd.it}
}

\maketitle

\begin{abstract}
Isoperimetric inequalities form a very intuitive yet powerful characterization of the connectedness of a state space, that has proven successful in obtaining convergence bounds. Since the seventies they form an essential tool in differential geometry \cite{cheeger1969,buser1982}, graph theory \cite{fiedler1973,alon1985} and Markov chain analysis \cite{dodziuk1984,aldous1987,lawler1988}. In this paper we use isoperimetric inequalities to construct a bound on the convergence time of any local probabilistic evolution that leaves its limit distribution invariant. We illustrate how this general result leads to new bounds on convergence times beyond the explicit Markovian setting, among others on quantum dynamics.
\end{abstract}

This paper is concerned with the discrete-time spreading of a distribution along the edges of a graph. In essence we establish that even by exploiting global information about the graph and allowing a very general use of this information, this spreading can still not be accelerated beyond the conductance bound. Before providing more ample context, we start with a motivating example ascribed to Eugenio Calabi, but which came to our attention through the seminal 1969 paper by Jeff Cheeger \cite{cheeger1969}. Whereas the original example concerns differential geometry, we will apply it to a graph setting. 

\begin{figure}[htb]
 \centering
   \def\svgwidth{.7\columnwidth}
   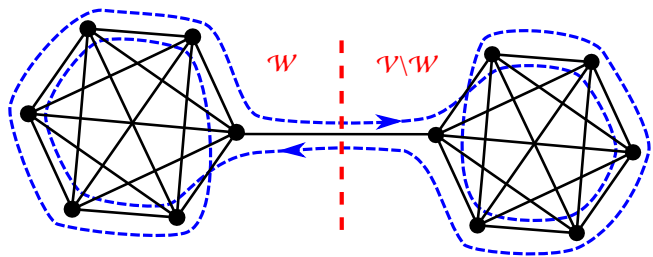
   \caption{(solid line) Dumbbell graph $K_n$-$K_n$ for $n=6$. (dashed line) superimposed cycle of length $4n$ in a construction towards faster mixing.}
   \label{fig:dumbbell}
\end{figure}

Consider a locality structure (discrete geometry) prescribed by the ``dumbbell'' graph family $K_n-K_n$ shown in Figure \ref{fig:dumbbell}, consisting of two complete graphs over $n$ nodes, connected by a single edge. The diameter of this graph, being the ``longest shortest path'' between any two nodes, is three. However, a random walk over this graph converging to the uniform distribution has an expected convergence time in $O(n^2)$. This convergence time can be improved with a ``global design'' but without violating locality of the evolution, by adding some memory to the walker. In Figure \ref{fig:dumbbell}, the system designer has superimposed a cycle (dashed line) over the dumbbell graph. By adding subnodes that allow to conditionally select different subflows through the graph (formally we ``lift'' the walk \cite{chen1999}), the walker can be restricted to walk along this cycle. Using this cycle, we can impose a strategy by Diaconis, Holmes and Neal \cite{diaconis2000,chen1999} to efficiently speed-up mixing over this cycle: \textit{let the walker cycle in the same direction with a probability $1-1/n$, and switch direction with probability $1/n$.} This way the walk will mix over the graph in $O(n)$, i.e.~quadratically faster than the original random walk. But this is still order $n$ times slower than the diameter. Nevertheless, we show in our paper that this improvement is the best possible for any local probabilistic process that leaves the target distribution invariant. So mixing in diameter time may be possible, but not without loosening any of these constraints.

\paragraph{\bf 1. Problem description and main result:}
Consider a graph $\cG$ with nodes $\cV$ and edges $\cE\subseteq \cV\times \cV$. We use the convention that $(i,i) \in \cE$ $\forall i \in \cV$. We define ``states'' $X$ as probability distributions over $\cV$. Given an initial state $X_0$, some system ``$\rightarrow$'' propagates it over $t$ time steps as $X_0\oset{t}{\to} X_t$. For a subset $\cW\subseteq \cV$ and a state $X$, we define $X(\cW)$ the probability of $\cW$ according to $X$, and $X|\cW$ as the state $X$ conditioned on being in $\cW$. We call $\cN(\cW)$ the neighborhood of $\cW\subseteq \cV$, i.e., the nodes outside $\cW$ that have an edge going to $\cW$. We impose the following fundamental properties.
 \begin{itemize}
  \item \textbf{linear initialization:} $\;\; X_0\oset{t}{\rightarrow} X_t,\; \tilde{X}_0\oset{t}{\rightarrow}\tilde{X}_t 
	   \;\;\Rightarrow\;\; pX_0 + (1\text{-}p)\tilde{X}_0 \oset{t}{\rightarrow} pX_t + (1\text{-}p)\tilde{X_t}$\vspace{2mm}
  \item \textbf{locality:} $\;\; \forall X_0, t\geq 0, \cW\subseteq \cV:\;\; X_{t+1}(\cW) \leq X_t(\cW) + X_t(\cN(\cW))$\vspace{2mm}
  \item \textbf{invariance:} 
$ X_0\oset{+\infty}{\rightarrow}\Pil\; \forall X_0 
		   \;\;\Rightarrow\;\; \Pil\oset{t}{\to} \Pil $
 \end{itemize}
The last property states that the unique steady state distribution of the system must be invariant as an initial condition.
The second property expresses that probability weight can only flow along an edge at each time, without referring to details of the system mapping ``$\rightarrow$''. The first property is natural as the input is a probability distribution. The point however is that the general process ``$\rightarrow$'' may e.g.~contain hidden states, and we here impose a linear initialization with the hidden states as well (see example below).

Our theorem presents a bound on the convergence of a system ``$\rightarrow$'' that obeys these conditions towards its steady state $\Pil$. Explicitly, let $\tau$ be a time step such that $\|X_{t}-\Pil\|_1 \leq 1/2$ for all $t \geq \tau$. In discrete geometry, given a graph $\cG$ and a limit distribution $\Pil$, the isoperimetric measure $\Phi$, which we also call the ``conductance'' \cite{aldous2002}, can be defined as:
 \[ \Phi = \max_P \Phi(P),\qquad \Phi(P) = \min_{\cW\subseteq \cV \, : \, \Pil(\cW)\leq 1/2} [P \circ (\Pil|\cW)] \, (\cV\setminus \cW). \]
The maximization is over all stochastic matrices $P$ acting on $\mathbb{R}^{|\cV|}$ that obey the locality of $\cG$ and for which $P \circ \Pil=\Pil$. In other words, ``$P \circ$'' is the most basic type of system ``$\rightarrow$'' satisfying our requirements: it is time-invariant and memoryless (``Markov''). If $\Pi$ is the uniform distribution, then $\Phi$ is upper bounded by the edge expansion of $\cG$, which is $1/n$ for the dumbbell graph. We establish the following ``conductance bound'' for \emph{any} more complicated system.
 \begin{theorem} \label{thm:main}
   If a system is linear, local and invariant, then
$\tau\geq 1/(8\Phi).$
 \end{theorem}
So for the dumbbell graph with $\Pi$ the uniform distribution, we find $\tau\geq n/8$ for any linear, local and invariant system.

Mixing on graph structures has drawn much interest for sampling algorithms, see e.g.~perfect matching \cite{jerrum1989}, or the Metropolis-Hastings algorithm used a lot in statistical mechanics. Bounds similar to Thm.1 have originally been proven by Cheeger \cite{cheeger1969} and Buser \cite{buser1982} in a differential geometry setting, and by Fiedler \cite{fiedler1973}, Dodziuk \cite{dodziuk1984} and Alon \cite{alon1985} in a discrete geometry and graph setting. In Markov chain analysis, early uses trace back to Aldous \cite{aldous1987}, Lawler and Sokal \cite{lawler1988} and Mihail \cite{mihail1989}. More recent examples are by Chen, Lov\'asz and Pak \cite{chen1999} who used a similar bound to prove that a restricted class of extended Markov chains called ``lifted Markov chains'' can at most quadratically accelerate convergence, and by Aharonov et al.~\cite{aharonov2001} to (loosely) bound the convergence speed of certain quantum processes.

Our result allows to improve known mixing bounds, e.g.~for quantum processes, and to generalize bounds beyond usual Markov chain settings, e.g.~by including nonlinear decision rules. Examples are briefly discussed after the proof.

\paragraph{\bf 2. Proof:} Our proof essentially comes down to two steps:
 \begin{itemize}
  \item \textit{locality implies particular simulability:} the locality condition implies that the dynamics can always be described using (time- and state-dependent) local stochastic matrices. This is not entirely trivial in such generality.
  \item \textit{bound for extended Markov chains:} we rather straightforwardly combine these matrices in an extended Markov chain model, for which we can prove the bound along standard lines.
 \end{itemize}

\paragraph{\bf 2.1. Locality implies simulability:}
A stochastic matrix $P$ is local if the system $X_0\to X_t=P^t \circ X_0$ is local. It is not hard to check that this coincides with the traditional definition, where locality means $P_{i,j}=0$ whenever $(i,j) \notin \cE$. The following Lemma kind of proves the converse. Its proof is inspired by a related result of Scott Aaronson \cite{aaronson2005}, establishing the lemma for quantum systems whose evolution is governed by a local unitary matrix.
\begin{lemma}\label{lem:simulable}
 If ``$\rightarrow$'' is a local system, then for every pair $(X_0,t)$ with $t>0$ there exists a local stochastic matrix $P(t,X_0)$ such that $X_0 \oset{t}{\rightarrow} X_t=P(t,X_0) \circ X_{t-1}$.
\end{lemma}
\begin{proof}
Call $Y=X_{t-1}$ and $Z=X_t$. We make a digression to \emph{flows over capacitated networks} \cite{ford1956} and consider the one shown in Figure \ref{fig:stoch-bridge}. The network consists of a source node $s$, a sink node $t$, and two copies of the graph nodes $\cV$ and $\cV'$. Node $s$ is connected to any node $i\in \cV$ with capacity $Y(i)$, any node $i\in \cV$ is connected with capacity 1 to any node $j\in \cV'$ iff $(i,j)\in \cE$ (else the nodes are not connected), and any node $j\in \cV'$ is connected to node $t$ with capacity $Z(j)$. If this network can route a steady flow of value 1 from node $s$ to node $t$, then the fraction of $Y(i)$ that is routed from $i \in \cV$ towards $j\in \cV'$ directly defines $P_{j,i}(t,X_0)$, as $Z(j) = \sum_{i\in\cV} P_{j,i}(t,X_0)Y(i)$ and so $P(t,X_0)\circ Y = Z$.

  \begin{figure}[htb]
    \centering
    \def\svgwidth{.65\columnwidth}
    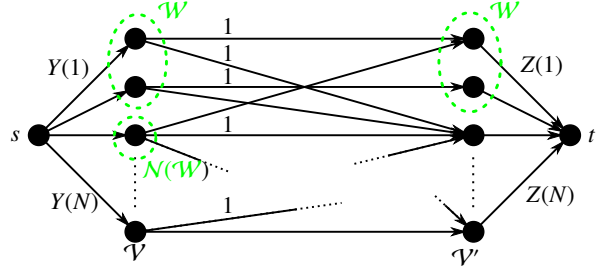
    \caption{Capacitated network construction used in Lemma \ref{lem:simulable}.}
    \label{fig:stoch-bridge}
   \end{figure}

The max-flow-min-cut theorem \cite{ford1956} states that the maximum steady flow which can be routed from node $s$ to node $t$ is equal to the minimum cut value of the graph, where a cut value is the sum of the capacities of a set of edges that disconnects $s$ from $t$.

It is clear that cutting all edges arriving at $t$ disconnects the graph, with a cut value of $1$, whereas cutting any middle edge between $\cV$ and $\cV'$ gives a cut value $\geq 1$. So the minimum cut need involve no such middle edge. Let us try to not cut the edges from $\cW \subseteq \cV'$ to $t$. To block any flow from $s$ to $t$ while keeping all middle edges, we must then cut the edges from $s$ to all the $l \in \cV$ which have an edge to $\cW$. This corresponds to all $l \in \cW \cup \cN(\cW)$. The value of this cut is thus
$${\textstyle 1-\sum_{j \in \cW} Z(j) + \sum_{j \in \cW} Y(j) + \sum_{j \in \cN(\cW)} Y(j) \; .}$$
Recalling that $Y=X_{t-1}$ and $Z=X_t$, locality imposes
$${\textstyle \sum_{j\in \cW}Z(j) \leq \sum_{j \in \cW} Y(j) + \sum_{j \in \cN(\cW)} Y(j) \, ,}$$
from which follows that the minimum cut value is $\geq 1$. According to the previous arguments, this concludes the proof. \qed
\end{proof}

\paragraph{\bf 2.2. Bound for ``extended'' Markov chains:} On the basis of these $P(t,X_0)$, we show how to construct a local Markov chain with at most twice the convergence time $\tau$ of our original system ``$\rightarrow$''. Thereto, we first define a closely related system
 \[ \oset{t}{\rightsquigarrow} \quad\equiv \quad \text{iterate } \oset{\tau}{\rightarrow} \quad \left( \text{namely floor} (t / \tau) \text{ times, plus } t\mathrm{mod}\tau \text{ steps of it} \right)\,. \]
By construction, ``$\rightsquigarrow$'' has the same convergence time $\tau$ as ``$\rightarrow$'', it has the same limit and it obeys the same locality and invariance conditions. We will now build a standard, time-invariant Markov chain that simulates the system ``$\rightsquigarrow$''.

To this end we first extend our state space: the original node set $\cV$ is lifted to $\hat{\cV}=(\cV;\{0,\dots,\tau-1\};\cV)$. From the perspective of a random walker, the first item contains its starting position, the second item a clock variable, and the last item its current position. In matrix form, with $\otimes$ representing the Kronecker product and $\cdot^\dagger$ the transpose, we now build the transition matrix $M$ for a Markov chain on $\hat{\cV}$ as follows:
 \[ M =  \sum_{i\in \cV}\sum_{t=0}^{\tau-2}e_ie_i^\dag \otimes e_{t+1}e_t^\dag \otimes P(t,e_i) 
		     + \sum_{i,j\in \cV} e_je_i^\dag\otimes e_0e_{\tau-1}^\dag\otimes e_je_j^\dag P(\tau-1,e_i). \]
Here $e_i$ is the unit vector with $1$ at index $i$, and $P(t,e_i)$ denotes the transition matrix obtained by Lemma 1 for $X_0 = e_i$, i.e., initial weight concentrated at node $i \in \cV$. This Markov chain simulates the ``$\rightsquigarrow$'' system in the following sense: 
when we locally initialize it in the state
 \[ v[X_0] = \sum_{i\in V} X_0(i) e_i\otimes e_0\otimes e_i\; , \]
the distribution over the subsets $(\cV;\{0,\dots,\tau-1\};i)$ of the resulting state $M^t v[X_0]$ at time $t$ exactly corresponds to $X_t$ resulting from $X_0 \rightsquigarrow X_t$.

A priori our Markov chain only simulates ``$\rightsquigarrow$'' for special initial states of the form $v[X_0]$ over $\hat{\cV}$. The following lemma shows that in fact when starting from an arbitrary distribution over $\hat{\cV}$, it takes at most twice the time to converge to $\|X_{t}-\Pil\|_1 \leq 1/2$ (over sets as just mentioned).
\begin{lemma} \label{lemma:conv-M}
If $\rightarrow$ has a convergence time $\tau$, then the Markov chain $M$ on $\hat{\cV}$ has a convergence time at most $2\tau$ over the subsets $\{\,(\cV;\{0,\dots,\tau-1\};i) : i \in \cV \,\}$.
\end{lemma}
\begin{proof}
Consider an arbitrary initial state $e_i\otimes e_T\otimes e_k$ for the Markov chain $M$. After $\tau-T$ steps, this state will necessarily have evolved to one of the special initial states of the form $v[X_0]$, for some $X_0$. By construction, the distribution of this state over the subsets $\{\,(\cV;\{0,\dots,\tau-1\};i) : i \in \cV \,\}$ will then simulate the evolution $X_0\oset{t}{\rightsquigarrow}X_t$, which converges to $\|X_{t}-\Pil\|_1 \leq 1/2$ for all $t\geq\tau$. Note indeed that $\rightsquigarrow$ and $\rightarrow$ are equivalent over the first $\tau$ time steps, and furthermore by invariance, iterating $\rightarrow$ via $\rightsquigarrow$ will never increase $\|X_{t}-\Pil\|_1$. Hence the Markov chain will have converged after $\tau-T+\tau\leq 2\tau$ steps at most. We have thus proved the convergence time for initial states with all weight concentrated on one element of $\hat{\cV}$. By linearity, this also proves the convergence time for arbitrary initial states. \qed
\end{proof}

The last element of the proof is a lower bound on $\tau$ for standard Markov chains. It essentially follows from a result by Chen, Lov\'asz and Pak \cite{chen1999}, stating that a class of extended Markov chains called ``lifted Markov chains'' that converge over $\hat{\cV}$ can converge at best in order $1/\Phi$, with $\Phi$ the conductance of the original graph. Our Markov chain $M$ does not exactly fit into this framework, because it is periodic on $\hat{\cV}$ and only its projection onto $\cV$ via the subsets of Lemma \eqref{lemma:conv-M} will converge. The proof can however be adapted to this case. Due to space constraints we must refer the reader to \cite{OurPrep} for a detailed proof, and we here only provide the statement:

\begin{lemma} \label{lemma:bound-M}
 The convergence time of $M$ over the sets $\{\,(\cV;\{0,\dots,\tau-1\};i) : i \in \cV \,\}$ is lower bounded by $1/(4\Phi)$.
\end{lemma}

By combining Lemma \ref{lemma:conv-M} and Lemma \ref{lemma:bound-M}, we obtain that $2\tau$ must be larger than $1/(4\Phi)$ and hence $\tau$ larger than $1/(8\Phi)$, as stated in the main theorem.

\paragraph{\bf 3. Examples:} We now discuss a few examples to illustrate the generality of our result. Note that the mathematical result is not restricted to cases where $X_t$ represents a probability distribution. It can apply to any situation where $X_t$ remains positive, bounded and preserves the sum of its components. Such dynamics can appear in flow dynamics and e.g.~average consensus algorithms for weight distribution \cite{TsitsiklisThesis}. In such settings our result might suggest how e.g.~relaxing the linearity constraint is necessary for beating the conductance bound.

\paragraph{\bf 3.1 Time-inhomogeneous Markov chains and Cesaro mixing:} The bound clearly includes time-varying Markov chains (that satisfy invariance), as appear in the proof. Practical examples of such processes can be found in \cite{saloffcoste2007}, and in \cite{Mossel} for card shuffling. The difficulty to analyze the convergence time of such processes is explicitly stated. Our paper thus provides a clear bound on the maximal achievable acceleration by exploiting the time-inhomogeneity degree of freedom in mixing algorithms.

Cesaro mixing \cite{levin2009} using a stochastic matrix $P$ is defined by the system $X_0 \oset{t}{\rightarrow} X_t = \frac{1}{t+1}\sum_{k=0}^{t}P^kX_0$. There appears to be no obvious way to write this as a Markov chain. However, one can show that Cesaro mixing satisfies our assumptions, so Thm.1 allows to directly bound the mixing time of such processes.

\paragraph{\bf 3.2 Processes with state-dependent or nonlocal decision rules:} The locality condition concerns how much probability weight is transferred at each step, but not how the decision about this transfer is taken. The latter is constrained by linearity. In this sense, our result directly bounds any attempts at adapting fast converging nonlinear algorithms e.g.~from consensus, towards truly probabilistic Markov chains where linearity is natural. 
Consider a nonlinear update rule from consensus, like:
$$X_{t+1} = P(Z_t)\, X_t \, .$$
In \cite{murray2003}, $Z_t$ is a static function of the weight differences on the respective links, e.g.~the weight associated to link $(i,j)$ in $P(Z)$ is a function of $X_t(i)-X_t(j)$. Our framework would even admit $Z_t$ being a dynamic function of $X$, possibly nonlinear, taking values in any space, and \emph{would not even require that it is based on local values of $X$ only:} e.g.~the weight associated to link $(i,j)$ in $P(Z)$ might be a function of some $X_t(k)$ where node $k$ is totally elsewhere in the graph.

Such update rule is in general not linear in $X_0$, but one may attempt to adapt it in this sense in the hope of designing e.g.~stochastic automata that improve mixing over standard Markov chains. For instance, one might imagine a system that distinguishes, in memory, each part of $X_t$ that has started from a different node at $X_0$. Once this is done, we are free to choose the evolution (possibly nonlinear, nonlocal) for each of these $X_0$-indexed parts, postulating that the full $X_t$ consists of their linear combination. One might thus wonder whether such heuristic approach could lead to faster mixing on e.g.~the dumbbell graph. Our result implies that --- provided also invariance is required --- such acceleration attempts are all limited to the conductance bound.

\paragraph{\bf 3.3 Finite-time convergence:} Consider the following algebraic problem, related to finite-time convergence \cite{hendrickx2014} and the inverse eigenvalue problem \cite{hogben2005}:
 \begin{quote}
  What is the minimal number of symmetric stochastic matrices over a graph $\cG$ whose product has all but one eigenvalue equal to zero?
 \end{quote}
From Theorem \ref{thm:main} it follows that this number is bounded by $1/(8\Phi)$. To see this, note that a set of local, symmetric, stochastic matrices $\{P(l),1\leq l\leq T\}$ over a graph $\cG$ defines a linear and local system by $X_0\oset{t}{\rightarrow} X_t = \left(\Pi_{l=1}^tP(l) \right)X_0$ for $t\leq T$. The system is also invariant as the matrix product leaves the all-ones vector $\vec{1}$ invariant: $\vec{1} \oset{t}{\rightarrow} \vec{1}$. If the product $\Pi_{l=1}^TP(l)$ has all but one eigenvalue equal to zero, the remaining eigenvalue necessarily being 1 with eigenvector $\vec{1}$, then necessarily the system has converged: $X_0 \oset{T}{\rightarrow} \vec{1}/\|X_0\|_1,\forall X_0$ and so $\tau\leq T$. By Theorem \ref{thm:main} the convergence of any linear, local and invariant system is bounded, specifically stating that $T\geq 1/(8\Phi)$.

\paragraph{\bf 3.4 Quantum walks:} The convergence properties of quantum processes spreading over localized state spaces play a role both in physics (e.g.~transport of excitations in photosynthesis \cite{mohseni2008}) and in quantum computation (e.g.~quantum random walks \cite{aharonov2001}). 

A discrete-time quantum walk is (although to our knowledge this has never been said explicitly) the generalization of a lifted walk, by keeping coherences among the node options. Denote by $\rho$ the quantum state, i.e.~a positive definite ``density matrix'' with trace one, whose diagonal represents probabilities over $\hat{\cV} = \{ (i,z) \}$ where $i\in \cV$ is a graph node and $z$ is a possible auxiliary degree of freedom (see introductory example \cite{diaconis2000}, coined quantum walks \cite{aharonov2001}, or Section 2.1). A general quantum walk follows:
$$\rho_{t+1} = \Psi \circ \rho_t$$
where $\Psi$ is a completely positive trace-preserving map; most popular is the unitary quantum walk, where $\Psi \circ \rho_t = U \rho_t U^\dagger$, with $U$ a unitary matrix satisfying the locality of the graph $\cG$, exactly as $P$ does for a Markov chain and $M$ does in Section 2.1. If $\rho_t$ would remain diagonal, this would correspond exactly to a lifted Markov chain \cite{chen1999}.
Authors have been wondering for some time whether the additional information contained in the off-diagonal elements of $\rho_t$ (``quantum coherences'') might allow faster mixing. In \cite{aharonov2001} a conductance bound is given for \emph{unitary} quantum walks and within a factor of the graph degree; but such factor becomes dominant for e.g.~the dumbbell graph. 

As will be further worked out in a future publication, quantum walks do satisfy the conditions of this paper, for most reasonable initializations. Then our result improves the bound of \cite{aharonov2001}, both by generalizing it to non-unitary walks and by getting rid of the degree-dependent factor. Quantum walks indeed satisfy locality, including the hidden (complex) variables representing coherences. Linearity trivially holds, except if one allows to initialize the walk with nonlocal coherences already. In other words, since~$\rho_0$ would necessarily be block-diagonal when all the initial weight is concentrated on a single node, introducing off-diagonal initial blocks when starting with a distribution over nodes would break linearity --- and by Thm.1, this would be necessary to potentially beat the conductance bound.

\end{document}

%% file: barbell.pdf_tex
\begingroup%
  \makeatletter%
  \providecommand\color[2][]{%
    \errmessage{(Inkscape) Color is used for the text in Inkscape, but the package 'color.sty' is not loaded}%
    \renewcommand\color[2][]{}%
  }%
  \providecommand\transparent[1]{%
    \errmessage{(Inkscape) Transparency is used (non-zero) for the text in Inkscape, but the package 'transparent.sty' is not loaded}%
    \renewcommand\transparent[1]{}%
  }%
  \providecommand\rotatebox[2]{#2}%
  \ifx\svgwidth\undefined%
    \setlength{\unitlength}{243.89035645bp}%
    \ifx\svgscale\undefined%
      \relax%
    \else%
      \setlength{\unitlength}{\unitlength * \real{\svgscale}}%
    \fi%
  \else%
    \setlength{\unitlength}{\svgwidth}%
  \fi%
  \global\let\svgwidth\undefined%
  \global\let\svgscale\undefined%
  \makeatother%
  \begin{picture}(1,0.38458814)%
    \put(0.40352583,0.28766491){\color[rgb]{0,0,0}\makebox(0,0)[lb]{\smash{$\textcolor{red}{\mathcal{W}}$}}}%
    \put(0.57387435,0.28613875){\color[rgb]{0,0,0}\makebox(0,0)[lb]{\smash{$\textcolor{red}{\mathcal{V}\backslash\mathcal{W}}$}}}%
    \put(0,0){\includegraphics[width=\unitlength,page=1]{barbell.pdf}}%
  \end{picture}%
\endgroup%

%% file: stoch-bridge.pdf_tex
\begingroup%
  \makeatletter%
  \providecommand\color[2][]{%
    \errmessage{(Inkscape) Color is used for the text in Inkscape, but the package 'color.sty' is not loaded}%
    \renewcommand\color[2][]{}%
  }%
  \providecommand\transparent[1]{%
    \errmessage{(Inkscape) Transparency is used (non-zero) for the text in Inkscape, but the package 'transparent.sty' is not loaded}%
    \renewcommand\transparent[1]{}%
  }%
  \providecommand\rotatebox[2]{#2}%
  \ifx\svgwidth\undefined%
    \setlength{\unitlength}{493.90244141bp}%
    \ifx\svgscale\undefined%
      \relax%
    \else%
      \setlength{\unitlength}{\unitlength * \real{\svgscale}}%
    \fi%
  \else%
    \setlength{\unitlength}{\svgwidth}%
  \fi%
  \global\let\svgwidth\undefined%
  \global\let\svgscale\undefined%
  \makeatother%
  \begin{picture}(1,0.42253806)%
    \put(0,0){\includegraphics[width=\unitlength,page=1]{stoch-bridge.pdf}}%
    \put(0.0125784,0.20196666){\color[rgb]{0,0,0}\makebox(0,0)[lb]{\smash{$s$}}}%
    \put(0.98210817,0.19997563){\color[rgb]{0,0,0}\makebox(0,0)[lb]{\smash{$t$}}}%
    \put(0.07425962,0.3123051){\color[rgb]{0,0,0}\makebox(0,0)[lb]{\smash{$Y(1)$}}}%
    \put(0.07986604,0.08776423){\color[rgb]{0,0,0}\makebox(0,0)[lb]{\smash{$Y(N)$}}}%
    \put(0,0){\includegraphics[width=\unitlength,page=2]{stoch-bridge.pdf}}%
    \put(0.8694606,0.31746215){\color[rgb]{0,0,0}\makebox(0,0)[lb]{\smash{$Z(1)$}}}%
    \put(0.87917911,0.09944348){\color[rgb]{0,0,0}\makebox(0,0)[lb]{\smash{$Z(N)$}}}%
    \put(0.36872521,0.38141927){\color[rgb]{0,0,0}\makebox(0,0)[lb]{\smash{1}}}%
    \put(0.36872521,0.33606618){\color[rgb]{0,0,0}\makebox(0,0)[lb]{\smash{1}}}%
    \put(0.36872521,0.29719211){\color[rgb]{0,0,0}\makebox(0,0)[lb]{\smash{1}}}%
    \put(0.36872521,0.21620445){\color[rgb]{0,0,0}\makebox(0,0)[lb]{\smash{1}}}%
    \put(0.36872521,0.08014519){\color[rgb]{0,0,0}\makebox(0,0)[lb]{\smash{1}}}%
    \put(0,0){\includegraphics[width=\unitlength,page=3]{stoch-bridge.pdf}}%
    \put(0.20375377,0.00138416){\color[rgb]{0,0,0}\makebox(0,0)[lb]{\smash{$\cV$}}}%
    \put(0.75401122,-0.00753561){\color[rgb]{0,0,0}\makebox(0,0)[lb]{\smash{$\cV'$}}}%
    \put(0.81809305,0.40774239){\color[rgb]{0,0,0}\makebox(0,0)[lb]{\smash{$\textcolor{green}{\mathcal{W}}$}}}%
    \put(0,0){\includegraphics[width=\unitlength,page=4]{stoch-bridge.pdf}}%
    \put(0.2612189,0.40730611){\color[rgb]{0,0,0}\makebox(0,0)[lb]{\smash{$\textcolor{green}{\mathcal{W}}$}}}%
    \put(0,0){\includegraphics[width=\unitlength,page=5]{stoch-bridge.pdf}}%
    \put(0.23530285,0.1416666){\color[rgb]{0,0,0}\makebox(0,0)[lb]{\smash{$\textcolor{green}{\mathcal{N}(\mathcal{W}})$}}}%
  \end{picture}%
\endgroup%